\documentclass{article}
\usepackage{graphicx} 
\usepackage{amsfonts}
\usepackage{amsthm,amssymb,latexsym,graphicx,amsmath,lineno}
\newtheorem{theorem}{Theorem}
\newtheorem{remark}[theorem]{Remark}
\newtheorem{lemma}[theorem]{Lemma}
\newtheorem{proposition}[theorem]{Proposition}

\title{A polynomial quantum computing algorithm for solving the dualization problem}
\author{Fernando Cuartero Gomez, Mauro Mezzini, Fernando Pelayo, \\
Jose Javier Paulet Gonzales,Hernan Indibil de la Cruz Calvo\\ Vicente Pascual  }

\date{September 2023}

\begin{document}

\maketitle

\begin{abstract}
Given two prime monotone boolean functions $f:\{0,1\}^n \to \{0,1\}$ and $g:\{0,1\}^n \to \{0,1\}$ the dualization problem consists in determining if $g$ is the dual of $f$, that is if $f(x_1, \dots, x_n)= \overline{g}(\overline{x_1},  \dots \overline{x_n})$ for all $(x_1, \dots x_n) \in \{0,1\}^n$. Associated to the dualization problem there is the corresponding decision problem: given two monotone prime boolean functions $f$ and $g$ is $g$ the dual of $f$? In this paper we present a quantum computing algorithm that solves the decision version of the dualization problem in polynomial time.
\end{abstract}

\section{Introduction}
A boolean function is monotone if  given any two boolean vectors $v=(v_1, \dots, v_n)$ and $w=(w_1, \dots, w_n)$ if $v_i \leq  w_i$ for all $i\in \{1, \dots , n\}$ we have that  $f(v) \leq f(w)$. 

The \emph{dualization} problem \cite{doi:10.1137/S009753970240639X, doi:10.1137/S0097539793250299, 10.1016/j.dam.2007.04.017, journals/jal/FredmanK96}, given a monotone boolean function $f:\{0,1\}^n \to \{0,1\}$ expressed in a prime (i.e. irredundant) disjunctive normal form (DNF), consists in finding the prime DNF of a monotone boolean function $g$ such that $f(x)=\overline {g}(\overline {x})$ for all $x \in \{0,1\}^n$. 
The decision version of the dualization problem, called \emph{dual}, is defined as follows: given two  prime monotone boolean functions $f$ and $g$ is $g$ the dual of $f$?
The dualization problem and its associated decision version, are prominent problems in  several research areas such as machine learning and data mining \cite{9efd059370054b5f98ecbe844e78cf04, doi:10.1137/S0097539700370072, 10.1007/3-540-45841-7_10, 10.1023/A:1007627028578} artificial intelligence \cite{k95, gps98, r87} and others (see \cite{doi:10.1137/S009753970240639X} and the references within ).
Borrowing the notation from  \cite{journals/jal/FredmanK96} we express the monotone boolean functions $f$ and $g$ in DNF as 
\[f= \bigvee_{I \in F} \bigwedge_{i \in I} x_i
\] and 
\[g=\bigvee_{J \in G} \bigwedge_ {j \in J} x_j
\]
where $I, J \subseteq \{1,2, \dots, n\}$ and $F$ (resp. $G$) is the set of prime implicants of $f$ (resp. $g$).
The best deterministic classical computing algorithm for solving the dual problem  has complexity $O(N^{o(\log N)})$ where $N$ is the number of prime implicants of $f$ and $g$, that is $N= |F|+|G|$ \cite{journals/jal/FredmanK96}.  Determining the complexity status of the dualization problem and its associated decision version is a prominent open problem.
Equally interesting is the self-dualization problem, that is, the problem of determining if a monotone boolean function is self-dual. It has the same complexity of the dual problem since it can be reduced to self-dualization of the function $yf \vee zg \vee yz$ where $y$ and $z$ are two additional boolean variables \cite{journals/jal/FredmanK96}.
In this paper we develop a polynomial time quantum computing algorithm for the dual (resp. self-dual) problem.

\section{Methods}
In the following the variable $x$ is interpreted sometimes as a boolean (or binary) $n$-dimensional vector and sometimes as the decimal expression of the binary vector. In particular if $x$ is the decimal value of the binary vector $(x_1, \dots, x_n)$ then the decimal value of the binary vector $(\overline x_1, \dots, \overline  x_n)$ is $\overline  x = 2^n-x-1$.
We start with the following propositions which will be  much used later in the paper. 

\begin{proposition} [\cite{journals/jal/FredmanK96}] \label{prop:intersection}
Necessary condition for two monotone boolean functions $g$ and $f$ expressed in their DNF to be mutually dual is that 
\begin{equation} \label{eq:intersect}
I \cap J \neq \emptyset  \text{ for every $I \in F$ and $J \in G$}
\end{equation}
\end{proposition}
\begin{proof}
If, by contradiction, there exist implicants $I \in F$ and $J \in G$ such that $I \cap J = \emptyset$, let $x= (x_1,\dots, x_n)$ such that $x_i=1$ if $i \in I$ and $x_i =0$ if $i \notin I$. Clearly $f(x)=1=g(\overline x)$ and $f$ and $g$ could not be mutually dual.
\end{proof}
By Proposition \ref{prop:intersection}, if $f$ is self-dual then every implicant of $F$ must intersect every other implicant. 

\begin{lemma} \label{balanced}
Suppose $f$ is self-dual. Then $f$ is balanced, that is, for half of $x$ values is 0 and for the other half is 1.
\end{lemma}
\begin{proof}
Let $0 \leq x <2^n$, then $\overline{x}= 2^n-x-1$. Furthermore since $f$ is self-dual we have that $f(x)\neq  f(\overline{x})$ for all $0 \leq x <2^n$. Therefore
\begin{align*}
\sum_{x=0}^{2^{n-1}-1}& f(x)+ \sum_{x=2^{n-1}}^{2^{n}-1} f( x ) =\\
\sum_{x=0}^{2^{n-1}-1}& f(x)+ \sum_{x=0}^{2^{n-1}-1} f( 2^n-x-1 ) =\\
\sum_{x=0}^{2^{n-1}-1}& [f(x)+  f( \overline x )] =2^{n-1}
\end{align*}
\end{proof}

\begin{lemma} \label{self-dual}
Let $f$ be a monotone boolean function expressed in its DNF which satisfies also \eqref{eq:intersect}. Then $f$ is self-dual if and only if $\sum_{x=0}^{2^{n}-1} f(x)= 2^{n-1}$
\end{lemma}
\begin{proof}
The necessity is given by Lemma \ref{balanced}. As for the sufficiency, suppose that  $\sum_{x=0}^{2^{n}-1} f(x)= 2^{n-1}$ and suppose by contradiction that $f(x)= f(\overline x)$ for some $0 \leq x < 2^{n}$. Since \eqref{eq:intersect} holds, when $f(x)=1$ there exists an implicant $I$ such that $x_i=1$ for all $i \in I$. But then $f(\overline x)=0$ since $I$ intersects all other implicants of $F$. In other words $f(x)+f(\overline x)\leq 1$ for all $x$. Therefore we must have that $f(z)=f(\overline z)=0$ for some $0\leq z <2^n$. But since 
\[
2^{n-1}= \sum_{x=0}^{2^n-1} f(x)=\sum_{x=0}^{2^{n-1}-1} [f(x)+f(\overline x)] \leq 2^{n-1}
\]
we must have, for every $0  \leq x <2^{n-1}$, that $f(x)+f(\overline x) = 1$, and this is a contradiction.
\end{proof}
We define $w(x)$ the \emph{Hamming weight} of the  integer $0\leq x <2^n$, as the number of ones in the binary representation of $x$, or, equivalently, if $x=(x_1, \dots, x_n)$ is a binary vector, then $w(x)= \sum_{i=1}^n x_i$.

We said that the complexity of the dualization problem is measured with respect to the combined size of $f$ and $g$, that is, with respect to $N = |F|+|G|$. Furthermore as stated in \cite{journals/jal/FredmanK96}, the number $n$ of variables of the boolean functions is always less than $|F||G|$. However there exists instances of the self-dual problem in which $N=O(2^n)$ as in the following example. 

Choose $n>4$ odd and consider the following boolean function $\varphi$ whose set of implicants $F$ is the set of all subsets of $\{1,\dots, n\}$ of cardinality $\left \lceil n/2 \right \rceil$ where $\lceil a \rceil $ is the least integer greater or equal than $a$.
\begin {lemma} \label{number_of_implicants}
The function $\varphi$ is self-dual and the number of its implicants is $\binom{n}{\left \lceil n/2 \right \rceil}$. 
\end{lemma}
\begin{proof}
Trivially $|F|= \binom{n}{\left \lceil n/2 \right \rceil}$. If there exist two implicants $I$ and $J$ such that $I \cap J= \emptyset$ then $|I \cup J|= |I|+|J|= 2 \left \lceil n/2 \right \rceil>n$ a contradiction to the fact that the number of variables is $n$. So we have that \eqref{eq:intersect} holds.

For every $x$ such that $w(x) < \left \lceil n/2 \right \rceil$ we have that $\varphi(x)=0$ since every implicant $I$ of $\varphi(x)$ has cardinality $|I|= \left \lceil n/2 \right \rceil$. On the other hand for every $x$ such that $w(x)\geq \left \lceil n/2 \right \rceil$ then $\varphi(x)=1$ since if we consider $x$ as a binary vector we will always find an implicant $I$ such that $x_i=1$ for all $i \in I$. Now it is immediate to check that $|\{x: w(x)\geq \left \lceil n/2\right \rceil|= 2^{n-1}$. By Lemma \ref{self-dual}, $\varphi$ is self-dual.
\end{proof}

\subsection{The quantum computing algorithm}

Given two boolean function $f$ and $g$ we build the function $h(x)= f(x) \oplus \overline{g}(\overline{x})$ where $\oplus$ is the sum modulo two.

Note that $h$ can be obtained from $f$  and $g$ by using a linear number of logic gates. If $f(x)=\overline{g}(\overline{x})$ for all $x$ then $h(x)=0$ for all $x$.
 We prepare a black box $U_h$ which performs the transformation $|x\rangle|y\rangle \to |x\rangle|y \oplus h(x)\rangle$, for $0 \leq x <2^{n }$. We use the blackbox in the Deutsch-Joshua algorithm. We have that the measurements of first $n$ qubits will be 
\[\displaystyle \dfrac{1}{2^n} \sum_{z=0}^{2^n-1}\sum_{x=0}^{2^n-1} (-1)^{x\cdot z+h(x)}|z \rangle
\]
and  the probability of measuring for $|z \rangle=|0 \rangle$ is, when $h(x)=0$ for all $x$, equal to 1 since 
\[\displaystyle \dfrac{1}{2^n} \sum_{x=0}^{2^n-1} (-1)^{h(x)}|0 \rangle=|0 \rangle
\]
so we have the following remark
\begin{remark}
Let $f$ and $g$ two monotone prime boolean functions and $h=f\oplus g$. If we measure at the end of the Deutsch-Joshua algorithm with blackbox function $h$,  a value $|x \rangle \neq |0 \rangle$ then $f$ is not the dual of $g$.
\end{remark}

From Remark 1 Lemma \ref{balanced} and Lemma \ref{self-dual} we can devise a simple quantum algorithm for checking if a function $f$ is self-dual as follows.\\\\
\textbf{Algorithm Quantum Dual}\\ 
\textbf{Input:} A black box $U_f$ which performs the transformation $|x\rangle|y\rangle \to |x\rangle|y \oplus f(x)\rangle$, for $0 \leq x <2^{n }$ and $f(x) \in \{0, 1\}$\\\\
\textbf{Output:} \emph{True} if $f$ is self-dual and \emph{False} otherwise.\\\\
\textbf{Procedure:} 
\begin{enumerate}
\item  Use the Deutsch-Joshua algorithm to check if $f$ is balanced. If the output of the Deutsch-Joshua algorithm is equal to $|0\rangle$ then output \emph{False} and exit.
\item Let $h(x)=f(x) \oplus \overline f (\overline x)$. Use the Deutsch-Joshua algorithm to check if $h$ is constant. If the output of the Deutsch-Joshua algorithm is not equal to $|0\rangle$ then output \emph{False} and exit.
\item Use the Quantum Counting algorithm to count the number of $x$ such that $f(x)=1$ using $t= \left \lceil n/2 \right \rceil$ qubits to measure the phase angle. If the measurement at the end of the algorithm is $|y\rangle$ and if $y \neq 2^{t-2}$ then output \emph{False} and exit.
\item Use the Grover algorithm to find an $x$ such that $f(x)=f(\overline x)$. If such $x$ is found then output \emph{False} and exit.
\item Output \emph{True}
\end{enumerate}

The complexity of the algorithm is dominated by the complexity of the Quantum Counting and of the Grover algorithms.  Both algorithms achieve a complexity on the number of quantum gates which is $O(2^{n/2})$ while the best deterministic classical computing algorithm has time complexity of $O(N^{o(\log N)})$ \cite{journals/jal/FredmanK96}. However, we saw in Lemma \ref{number_of_implicants} that a self-dual function can have a number of implicants in its DNF equal to $\binom{n}{\left \lceil n/2 \right \rceil}$ which is asymptotic to $O(2^n)$. Therefore we have that $N \leq 2^n$ from which we obtain that the complexity of our quantum algorithm for the dualization problem  is $O(\sqrt{N})$. 

\bibliographystyle{elsarticle-num}
\bibliography{notes_on_HSHx}

\end{document}